\documentclass[11pt, reqno]{amsart}%
\usepackage{graphics,color}
\usepackage{amsfonts}
\usepackage{amsmath}
\usepackage{amssymb}
\usepackage{graphicx}
\usepackage[bookmarks=false]{hyperref}%
\newtheorem{thm}{Theorem}[section]
\newtheorem{theorem}{Theorem}[section]
\newtheorem{lemma}[thm]{Lemma}
\newtheorem{cor}[thm]{Corollary}
\newtheorem{conj}[thm]{Conjecture}
\newtheorem{prop}[thm]{Proposition}

\newtheorem{defn}[thm]{Definition}
\newtheorem{example}[thm]{Example}
\newtheorem{remark}[thm]{Remark}

\numberwithin{equation}{section}
\begin{document}
\title[Nested Canalyzing Functions And Their Average Sensitivities]{Nested Canalyzing Functions And Their Average Sensitivities}
\author[Yuan Li, John O. Adeyeye, Reinhard Laubenbacher]{Yuan Li$^{1\ast}$, John O. Adeyeye $^{2\ast}$, Reinhard Laubenbacher$^{3}$}
\address{{\small $^{1}$Department of Mathematics, Winston-Salem State University, NC
27110,USA}\\
{\small email: liyu@wssu.edu }\\
$^{2}$Department of Mathematics, Winston-Salem State University, NC 27110,USA,
{\small email: adeyeyej@wssu.edu}\\
{\small $^{3}$Virginia Bioinformatics Institute, Virginia Tech, Blacksburg, VA
24061,USA }\\
{\small email: reinhard@vbi.vt.edu}}
\thanks{$^{\ast}$ Supported by an award from the USA DoD $\#$ W911NF-11-10166}
\keywords{Nested Canalyzing Function, Layer Number, Extended Monomial, Multinomial Coefficient,
Dynamical System, Hamming Weight, Activity, Average Sensitivity. }
\date{}

\begin{abstract}
In this paper, we obtain complete characterization for nested canalyzing
functions (NCFs) by obtaining its unique algebraic normal form (polynomial
form). We introduce a new concept, LAYER NUMBER for NCF. Based on this,  we obtain  explicit formulas for
the the following important parameters: 1) Number of all the nested canalyzing functions, 2) Number of all the NCFs with given LAYER NUMBER, 3) Hamming weight of any NCF, 4) The activity number of any variable of any NCF, 5) The average sensitivity of any NCF.
Based on these formulas, we show the activity number is greater for those variables in out layer and equal in the same layer. We show the average sensitivity attains minimal value when the NCF has only one layer. We also prove the average sensitivity for any NCF (No matter how many variables it has)  is between $0$ and $2$. Hence, theoretically, we show why NCF is stable since a random Boolean function has average sensitivity $\frac{n}{2}$. Finally we conjecture that the NCF attain the maximal average sensitivity if it has the maximal LAYER NUMBER $n-1$. Hence, we guess the uniform upper bound for the average sensitivity of any NCF can be reduced to $\frac{4}{3}$ which is tight.

\end{abstract}
\maketitle

\section{Introduction}

\label{sec-intro} Canalyzing function were introduced by Kauffman \cite{Kau1}
as appropriate rules in Boolean network models or gene regulatory networks.
Canalyzing functions are known to have other important applications in
physics, engineering and biology. In \cite{Mor} it was shown that the dynamics
of a Boolean network which operates according to canalyzing rules is robust
with regard to small perturbations. In \cite{Win2}, W. Just, I. Shmulevich and J. Konvalina 
derived an exact formula for the number of canalyzing functions. In \cite{Yua2}, 
the definition of canalyzing functions was generalized to any finite fields $\mathbb{F}_{q}$,
where $q$ is a power of a prime.  Both the exact formulas and the
asymptotes of the number of the generalized canalyzing functions were obtained.

Nested Canalyzing Functions (NCFs) were introduced recently in \cite{Kau2}.
One important characteristic of (nested) canalyzing functions is that they
exhibit a stabilizing effect on the dynamics of a system. That is, small
perturbations of an initial state should not grow in time and must eventually
end up in the same attractor of the initial state. The stability is typically
measured using so-called Derrida plots which monitor the Hamming distance
between a random initial state and its perturbed state as both evolve over
time. If the Hamming distance decreases over time, the system is considered
stable. The slope of the Derrida curve is used as a numerical measure of
stability. Roughly speaking, the phase space of a stable system has few
components and the limit cycle of each component is short.

In \cite{Kau3}, the authors studied the dynamics of nested canalyzing Boolean
networks over a variety of dependency graphs. That is, for a given random
graph on $n$ nodes, where the in-degree of each node is chosen at random
between $0$ and $k$, where $k\leq n$, a nested canalyzing function is assigned
to each node in terms of the in-degree variables of that node. The dynamics of
these networks were then analyzed and the stability measured using Derrida
plots. It is shown that nested canalyzing networks are remarkably stable
regardless of the in-degree distribution and that the stability increases as
the average number of inputs of each node increases.

An extensive analysis of available biological data on gene regulations (about
150 genes) showed that 139 of them are regulated by canalyzing functions
\cite{Har}. In \cite{Kau3, Nik}, it was shown that 133 of the 139 are in fact
nested canalyzing.

Most published molecular networks are given in the form of a wiring diagram,
or dependency graph, constructed from experiments and prior published
knowledge. However, for most of the molecular species in the network, little
knowledge, if any, could be deduced about their regulatory mechanisms, for
instance in the gene transcription networks in yeast \cite{Herr} and E. Coli
\cite{Bar}. Each one of these networks contains more than 1000 genes. Kauffman
et. al \cite{Kau2} investigated the effect of the topology of a sub-network of
the yeast transcriptional network where many of the transcriptional rules are
not known. They generated ensembles of different models where all models have
the same dependency graph. Their heuristic results imply that the dynamics of
those models which used only nested canalyzing functions were far more stable
than the randomly generated models. Since it is already established that the
yeast transcriptional network is stable, this suggests that the unknown
interaction rules are very likely nested canalyzing functions. In a recent
article \cite{Bal}, the whole transcriptional network of yeast, which has 3459
genes as well as the transcriptional networks of E. Coli (1481 genes) and B.
subtillis (840 genes) have been analyzed in a similar fashion, with similar findings.

These heuristic and statistical results show that the class of nested
canalyzing functions is very important in systems biology. It is shown in
\cite{Jar} that this class is identical to the class of so-called unate
cascade Boolean functions, which has been studied extensively in engineering
and computer science. It was shown in \cite{But} that this class produces the
binary decision diagrams with the shortest average path length. Thus, a more
detailed mathematical study of this class of functions has applications to
problems in engineering as well.

In \cite{Abd2},  the authors provided a
description of nested canalyzing function. As a corollary of the equivalence,
a formula in the literature for the number of unate cascade functions also
provides such a formula the number of nested canalyzing functions. Recently,
in \cite{Mur2}, those results were generalized to the multi-state nested
canalyzing functions on finite fields $\mathbb{F}_{p}$, where $p$ is a prime.
They obtained the formula for the number of the generalized NCFs, as a
recursive relation.

In \cite{Coo}, Cook et al. introduced the notion of sensitivity as a combinatorial 
measure for Boolean functions providing lower bounds on the time needed by CREW PRAM (concurrent read
 , but exclusive write (CREW) parallel random access machine (PRAM)). It was extended by Nisan \cite{Nis} to block
 sensitivity. It is still open whether sensitivity and block sensitivity are polynomially related (they are equal for monotone Boolean 
 functions). Although the definition is straightforward, the sensitivity is understood only for a few classes function. For monotone functions, Ilya Shmulevich \cite {Shm2} derived asymptotic formulas for a typical monotone Boolean functions. Recently, Shengyu Zhang \cite{Zha} find a formula  for the average sensitivity of any monotone Boolean functions, hence, a tight bound is derived.
 In \cite{Shm}, Ilya Shmulevich and Stauart A. Kauffman considered the activities of the variables of Boolean functions with only one 
 canalyzing variable. They obtained the average sensitivity of this kind of Boolean function.

In this paper, we revisit the NCF, obtaining a more explicit characterization
of the Boolean NCFs than those in \cite{Abd2}. We introduce a new concept, the
$LAYER$ $NUMBER$ in order to classify all the  variables. Hence, the
dominance of the variable can be quantified. As a consequence, we
obtain an explicit formula for the number of NCFs. Thus, a nonlinear recursive
relation (the original formula) is solved, which maybe of independent
mathematical interest.
Using our unique algebraic normal form of NCF, for any NCF, we get the formula of activity for its variables.
We show that the variables in a more dominant layer have greater activity number. Variables in the same layer have the same activity 
numbers. 
Consequently, we obtain the formula of any NCF's average sensitivity, its lower bound is $\frac{n}{2^{n-1}}$ and its upper bound is $2$ (No matter what $n$ is) which is much less
than $\frac{n}{2}$, the average sensitivity of a random Boolean function. So, theoretically, we proved why NCF is ``stable''. We also find the formula of  the Hamming weight of each NCF. Finally, we conjecture that the NCF attains its maximal value if it has the maximal LAYER NUMBER $n-1$. Hence, we guess the tight upper bound is $\frac{4}{3}$.
In the next section, we introduce some definitions and notations.

\section{Preliminaries}

\label{2} In this section we introduce the definitions and notations. Let
$\mathbb{F}=\mathbb{F}_{2}$ be the Galois field with $2$ elements. If $f$ is a
$n$ variable function from $\mathbb{F}^{n}$ to $\mathbb{F}$, it is well known
\cite{Lid} that $f$ can be expressed as a polynomial, called the algebraic
normal form(ANF):
\[
f(x_{1},x_{2},\ldots,x_{n})=\bigoplus_{0\leq k_i\leq 1,i=1,\ldots,n}a_{k_{1}k_{2}\ldots k_{n}}{x_{1}}^{k_{1}}{x_{2}}^{k_{2}%
}\cdots{x_{n}}^{k_{n}}%
\]
where each coefficient $a_{k_{1}k_{2}\ldots k_{n}}\in\mathbb{F}$ is a
constant. The number $k_{1}+k_{2}+\cdots+k_{n}$ is the multivariate degree of
the term $a_{k_{1}k_{2}\ldots k_{n}}{x_{1}}^{k_{1}}{x_{2}}^{k_{2}}\cdots
{x_{n}}^{k_{n}}$ with nonzero coefficient $a_{k_{1}k_{2}\ldots k_{n}}$. The
greatest degree of all the terms of $f$ is called the algebraic degree,
denoted by $deg(f)$.

\begin{defn}
\label{def2.1} $f(x_{1},x_{2},\ldots, x_{n})$ is essential in variable $x_{i}$
if there exist $r, s\in\mathbb{F}$ and $x_{1}^{*},\ldots,x_{i-1}^{*}$
$,x_{i+1}^{*},\ldots, x_{n}^{*}$ such that $f(x_{1}^{*},\ldots,x_{i-1}%
^{*},r,x_{i+1}^{*},\ldots,x_{n}^{*})\neq f(x_{1}^{*},\ldots,x_{i-1}%
^{*},s,x_{i+1}^{*},\ldots,x_{n}^{*})$.
\end{defn}

\begin{defn}
\label{def2.2} A function $f(x_{1},x_{2},\ldots,x_{n})$ is $<i:a:b>$
canalyzing if

$f(x_{1},\ldots,x_{i-1},a,x_{i+1},\ldots,x_{n})=b$, for all $x_{j}$, $j\neq
i$, where $i\in\{1,\dots,n\}$, $a$,$b\in\mathbb{F}$.
\end{defn}

The definition is reminiscent of the concept of "canalisation" introduced by
the geneticist C. H. Waddington \cite{Wad} to represent the ability of a
genotype to produce the same phenotype regardless of environmental variability.

\begin{defn}
\label{def2.3} Let $f$ be a Boolean function in $n$ variables. Let $\sigma$ be
a permutation on $\{1,2,\ldots,n\}$. The function $f$ is nested canalyzing
function (NCF) in the variable order

$x_{\sigma(1)},\ldots,x_{\sigma(n)}$ with canalyzing input values
$a_{1},\ldots,a_{n}$ and canalyzed values $b_{1},\ldots,b_{n}$, if it can be
represented in the form

$f(x_{1},\ldots,x_{n})=\left\{
\begin{array}
[c]{ll}%
b_{1} & x_{\sigma(1)}=a_{1},\\
b_{2} & x_{\sigma(1)}= \overline{ a_{1}}, x_{\sigma(2)}=a_{2},\\
b_{3} & x_{\sigma(1)}= \overline{ a_{1}}, x_{\sigma(2)}= \overline{ a_{2}},
x_{\sigma(3)}=a_{3},\\
.\ldots. & \\
b_{n} & x_{\sigma(1)}= \overline{ a_{1}}, x_{\sigma(2)}= \overline{ a_{2}%
},\ldots,x_{\sigma(n-1)}= \overline{ a_{n-1}}, x_{\sigma(n)}=a_{n},\\
\overline{b_{n}} & x_{\sigma(1)}= \overline{ a_{1}}, x_{\sigma(2)}= \overline{
a_{2}},\ldots,x_{\sigma(n-1)}= \overline{ a_{n-1}}, x_{\sigma(n)}=\overline{
a_{n}}.
\end{array}
\right. $

Where $\overline{a}=a\oplus 1$.The function f is nested canalyzing if f is nested
canalyzing in the variable order $x_{\sigma(1)},\ldots,x_{\sigma(n)}$ for some
permutation $\sigma$.
\end{defn}

Let $\alpha=(a_{1},a_{2},\ldots,a_{n})$ and $\beta=(b_{1},b_{2},\ldots,b_{n}%
)$, we say $f$ is $\{\sigma:\alpha:\beta\}$ NCF if it is NCF in the variable
order $x_{\sigma(1)},\ldots,x_{\sigma(n)}$ with canalyzing input values
$\alpha=(a_{1},\ldots,a_{n})$ and canalyzed values $\beta=(b_{1},\ldots
,b_{n})$.

Given vector $\alpha=(a_{1},a_{2},\ldots,a_{n})$, we define $\alpha
^{i_{1},\ldots,i_{k}}=(a_{1},\ldots,\overline{a_{i_{1}}},\ldots,\overline
{a_{i_{k}}},\ldots,a_{n})$

From the above definition, we immediately have the following

\begin{prop}
$f$ is $\{\sigma:\alpha:\beta\}$ NCF  $\Longleftrightarrow$ $f$ is
$\{\sigma:\alpha^{n}:\beta^{n}\}$ NCF
\end{prop}

\begin{example}
\label{exa2.1} $f(x_{1},x_{2},x_{3})=x_{1}(x_{2}\oplus 1)x_{3}\oplus 1$ is
$\{(1,2,3):(0,1,0):(1,1,1)\}$ NCF.

Actually, one can check this function is nested canalyzing in any variable order.
\end{example}

\begin{example}
\label{exa2.2} $f(x_{1},x_{2},x_{3})=(x_{1}\oplus 1)(x_{2}(x_{3}\oplus 1)\oplus 1)\oplus 1$. This
function is

$\{(1,2,3):(1,0,1):(1,0,0)\}$ NCF. It is also $\{(1,3,2):(1,1,1):(1,0,1)\}$ NCF.

One can check this function can be nested canalyzing in only two variable
orders $(x_{1},x_{2},x_{3})$ and $(x_{1},x_{3},x_{2})$.
\end{example}

From the above definitions, we know a function is NCF, all the $n$ variable
must be essential. However, a constant function $b$ can be  $<i:a:b>$
canalyzing for any $i$ and $a$.

\section{A Complete Characterization for NCF}

\label{3}

In \cite{Lor}, the author introduced Partially Nested Canalyzing Functions
(PNCFs), a generalization of the NCFs, and the nested canalyzing depth, which
measures the extent to which it retains a nested canalyzing structure. In
\cite{Win}, the author introduced the extended monomial system.

As we will see, in a Nested Canalyzing Function, some variables are more
dominant than the others. We will classify all the variables of a NCF into
different levels according to the extent of their dominance. Hence, we will
give description about NCF with more detail. Actually, we will obtain clearer
description about NCF by introducing a new concept: LAYER NUMBER. As a by-product, we also obtain some enumeration
results. Eventually, we will find an explicit formula of the number of all the NCFs.

First, we have

\begin{defn}
\label{def3.1} \cite{Win} $M(x_{1},\ldots,x_{n})$ is an extended monomial of
essential variables $x_{1},\ldots,x_{n}$ if $M(x_{1},\ldots,x_{n}%
)=(x_{1}\oplus a_{1})(x_{2}\oplus a_{2})...(x_{n}\oplus a_{n})$, where $a_{i}\in\mathbb{F}_{2}$.
\end{defn}

Basically, we will rewrite Theorem 3.1 in \cite{Abd2} with more information.

\begin{lemma}
\label{lm3.1}  $f(x_{1},x_{2},...x_{n})$ is $<i:a:b>$ canalyzing iff

$f(X)=f(x_{1},x_{2},...,x_{n})=(x_{i}\oplus a)Q(x_{1},\ldots, x_{i-1},x_{i+1}\ldots
x_{n})\oplus b$.
\end{lemma}

\begin{proof}
From the algebraic normal form of $f$, we rewrite it as $f=x_{i}g_{1}%
(X_{i})\oplus g_{0}(X_{i})$,  where $X_{i}=(x_{1},\ldots,x_{i-1},x_{i+1}%
,\ldots,x_{n})$. Hence,  $f(X)=f(x_{1},x_{2},\ldots,x_{n})$ $=(x_{i}%
\oplus a)g_{1}(X_{i})\oplus ag_{1}(X_{i})\oplus g_{0}(X_{i})$. Let $g_{1}(X_{i})=Q(x_{1},\ldots,
x_{i-1},x_{i+1}\ldots x_{n})$ and $r(X_{i})=ag_{1}(X_{i})\oplus g_{0}(X_{i})$. Then
$f(X)=f(x_{1},\ldots,x_{n})=(x_{i}\oplus a)Q(x_{1},\ldots, x_{i-1},x_{i+1}\ldots
x_{n})\oplus r(X_{i})$

Since $f(X)$ is $<i:a:b>$ canalyzing, we get $f(X)=f(x_{1},...x_{i-1}%
,a,x_{i+1},\ldots,x_{n})=b$ for any $x_{1},\ldots, x_{i-1},x_{i+1}\ldots
x_{n}$, i.e., $r(X_{i})=b$ for any $X_{i}$. So $r(X_{i})$ must be the constant
$b$. We finished the necessity. The sufficiency is obvious.
\end{proof}

\begin{remark}\label{remark1}
\label{re1} 1) When we contrast this lemma to the first part of Theorem 3.1 in
\cite{Abd2},we make clear that here,  the $x_{i}$ is not essential in $Q$. 2)
In \cite{Yua2}, there is a general version of this Lemma over any finite fields. 
3) In the above lemma, if $f$ is constant, then
$Q=0$.
\end{remark}

From Definition \ref{def2.3}, we have the following

\begin{prop}
\label{prop3.1} If $f(x_{1},\ldots,x_{n})$ is $\{\sigma:\alpha:\beta\}$ NCF,
i.e., if it is NCF in the variable order

$x_{\sigma(1)},\ldots,x_{\sigma(n)}$ with canalyzing input values
$\alpha=(a_{1},\ldots,a_{n})$ and canalyzed values $\beta=(b_{1},\ldots
,b_{n})$.

Then, for $1\leq k\leq n-1$, let $x_{\sigma(1)}=\overline{a_{1}}%
,\ldots,x_{\sigma(k)}=\overline{a_{k}}$, then the function

$f(x_{1},\ldots,\overset{\sigma(1)}{\overline{a_{1}}},\ldots,\overset
{\sigma(k)}{\overline{a_{k}}},\ldots, x_{n})$ is $\{\sigma^{*}:\alpha
^{*}:\beta^{*}\}$ NCF on those remaining variables, where $\sigma^{*}
=x_{\sigma(k+1)},\ldots,x_{\sigma(n)}$, $\alpha^{*}=(a_{k+1},\ldots,a_{n})$
and $\beta^{*}=(b_{k+1},\ldots,b_{n})$.
\end{prop}

\begin{defn}
\label{def3.2} If $f(x_{1},\ldots,x_{n})$ is a NCF. We call variable $x_{i}$
the most dominant  variable of $f$, if there is an order
$\alpha=(x_{i},\ldots)$ such that $f$ is NCF with this variable order(In other
words, if $f$ is also $<i:a:b>$ canalyzing for some $a$ and $b$).
\end{defn}

In Example \ref{exa2.1}, all the three variables are most dominant, in Example
\ref{exa2.2}, only $x_{1}$ is the most dominant  variable. We have

\begin{theorem}\label{th3.1} 
Given NCF $f(x_{1},\ldots,x_{n})$, all the variables are most
dominant iff

$f=M(x_{1},\ldots,x_{n})\oplus b$, where $M$ is an extended monomial, i.e.,

$M=(x_{1}\oplus a_{1})(x_{2}\oplus a_{2})...(x_{n}\oplus a_{n})$.
\end{theorem}

\begin{proof}
$x_{1}$ is the most dominant, from Lemma \ref{lm3.1}, we know there exist
$a_{1}$ and $b$ such that

$f(x_{1},x_{2},\ldots,x_{n})=(x_{1}\oplus a_{1})Q(x_{2},\ldots, x_{n})\oplus b$, i.e.,
$(x_{1}\oplus a_{1})|(f\oplus b)$. Now, $x_{2}$ is also the most dominant, we have $a_{2}$
and $b^{\prime}$ such that

$f(x_{1},a_{2},x_{3},\ldots,x_{n})=b^{\prime}$ for any $x_{1},x_{3}%
,\ldots,x_{n}$. Specifically, let $x_{1}=a_{1}$, we get

$f(a_{1},a_{2},x_{3},\ldots,x_{n})=b=b^{\prime}$. Hence, we also get
$(x_{2}\oplus a_{2})|(f\oplus b)=(x_{1}\oplus a_{1})Q(x_{2},\ldots, x_{n})$, since $x_{1}\oplus a_{1}$
and $x_{2}\oplus a_{2}$ are coprime, we get $(x_{2}\oplus a_{2})|Q(x_{2},\ldots, x_{n})$,
hence, $f(x_{1},x_{2},\ldots,x_{n})=(x_{1}\oplus a_{1})(x_{2}\oplus a_{2})Q^{\prime}%
(x_{3},\ldots, x_{n})\oplus b$. With induction principle, the necessity is proved.
The sufficiency if evident.
\end{proof}

We are ready to prove the following main result of this section.

\begin{theorem}\label{th2}
\label{th3.2} Given $n\geq2$, $f(x_{1},x_{2},\ldots,x_{n})$ is nested
canalyzing iff it can be uniquely written as
\begin{equation}\label{eq3.1}
f(x_{1},x_{2},\ldots,x_{n})=M_{1}(M_{2}(\ldots(M_{r-1}%
(M_{r}\oplus 1)\oplus 1)\ldots)\oplus 1)\oplus b.
\end{equation}
Where each $M_{i}$ is an extended monomial of a set of disjoint variables.
More precisely, $M_{i}=\prod_{j=1}^{k_{i}}(x_{i_{j}}\oplus a_{i_{j}})$,
$i=1,\ldots,r$, $k_{i}\geq1$ for $i=1,\ldots,r-1$, $k_{r}\geq2$, $k_{1}%
\oplus \ldots \oplus k_{r}=n$, $a_{i_{j}}\in\mathbb{F}_{2}$, $\{i_{j}|j=1,\ldots,k_{i},
i=1,\ldots,r\}=\{1,\ldots,n\}$.
\end{theorem}

\begin{proof}
We use induction on $n$.

When $n=2$, there are 16 boolean functions, 8 of them are NCFs, Namely

$(x_{1}\oplus a_{1})(x_{2}\oplus a_{2})\oplus c=M_{1}\oplus 1\oplus b$, where $b=1\oplus c$ and $M_{1}%
=(x_{1}\oplus a_{1})(x_{2}\oplus a_{2})$.

If $(x_{1}\oplus a_{1})(x_{2}\oplus a_{2})\oplus c=(x_{1}\oplus{a_{1}}^{\prime})(x_{2}\oplus{a_{2}%
)}^{\prime}\oplus c^{\prime}$, by equating the coefficients, we immediately obtain
$a_{1}={a_{1}}^{\prime}$, $a_{2}={a_{2}}^{\prime}$ and $c=c^{\prime}$. So,
uniqueness is true.

We have proved that equation \ref{eq3.1} is true for $n=2$, where $r=1$.

Let's assume that equation \ref{eq3.1} is true for any nested canalyzing
function which has at most $n-1$ essential variables.

Now, consider NCF $f(x_{1},\ldots,x_{n})$.

Suppose $x_{\sigma(1)},\ldots,x_{\sigma(k_{1})}$ are all the most dominant
canalyzing variables of $f$, $1\leq k_{1}\leq n$.

Case 1: $k_{1}=n$, by Theorem \ref{th3.1}, the conclusion is true with $r=1$.

Case 2: $k_{1}<n$, with the same arguments to Theorem \ref{th3.1}, we can get
$f=M_{1}g\oplus b$, where

$M_{1}=(x_{\sigma(1)}\oplus a_{\sigma(1)})\ldots(x_{\sigma(k)}\oplus a_{\sigma(k)})$. Let
$x_{\sigma(1)}=\overline{a_{\sigma(1)}},\ldots, x_{\sigma(k)}=\overline
{a_{\sigma(k)}}$ in $f$, the function $g\oplus b$, hence, $g$, of the remaining
variables will also be nested canalyzing by Proposition \ref{prop3.1}. Since
$g$ has $n-k_{1}\leq n-1$ variables, by induction assumption, we get

$g=M_{2}(M_{3}(\ldots(M_{r-1}(M_{r}\oplus 1)\oplus 1)\ldots)\oplus 1)\oplus b_{1}$, at this time,
$b_{1}$ must be $1$. Otherwise, all the variables in $M_{2}$ will also be the
most dominant variables of $f$. Hence, we are done.
\end{proof}

Because each NCF can be uniquely written as \ref{eq3.1} and the number $r$ is
uniquely determined by $f$, we have

\begin{defn}
\label{def3.3} For a NCF written as equation \ref{eq3.1}, the number $r$ will
be called its LAYER NUMBER. Essential variables of $M_{1}$ will be called the
most dominant variables(canalyzing variable), they belong to the first layer of this NCF.
Essential variables of $M_{2}$ will be
called the second most dominant variables and belong to the second layer 
of this NCF and etc.
\end{defn}

The function in example \ref{exa2.1} has LAYER NUMBER 1 and the function in
example \ref{exa2.2} has LAYER NUMBER 2.

\begin{remark}\label{remark2}
In Theorem \ref{th2}, 1) $k_r\geq 2$. It is impossible that $k_r=1$. Otherwise, $M_r\oplus 1$ will be a factor of $M_{r-1}$ which means LAYER NUMBER is $r-1$. 
2) If variable $x_i$ is in the first layer, and $x_i\oplus a_i$ is a factor of $M_i$, then this NCF is $<i:a_i:b>$ canalyzing, we simply say $x_i$ is a canalyzing variable of this NCF.
\end{remark}

Let $\mathbb{NCF}(n,r)$ stands for the set of all the $n$ variable nested
canalyzing functions with LAYER NUMBER $r$ and $\mathbb{NCF}(n)$ stands for
the set of all the $n$ variable nested canalyzing functions. We have

\begin{cor}
\label{co3.1} Given $n\geq2$,
\[
|\mathbb{NCF}(n,r)|=2^{n+1}\sum_{\substack{k_{1}+\ldots+k_{r}=n\\k_{i}%
\geq1,i=1,\ldots,r-1, k_{r}\geq2}}\binom{n}{k_{1},\ldots,k_{r-1}}%
\]
and
\[
|\mathbb{NCF}(n)|=2^{n+1}\sum_{\substack{r=1}}^{n-1}\sum_{\substack{k_{1}%
+\ldots+k_{r}=n\\k_{i}\geq1,i=1,\ldots,r-1, k_{r}\geq2}}\binom{n}{k_{1}%
,\ldots,k_{r-1}}%
\]
Where the multinomial coefficient $\binom{n}{k_{1},\ldots,k_{r-1}}=\frac
{n!}{k_{1}!\ldots k_{r}!}$
\end{cor}

\begin{proof}
From Equation \ref{eq3.1}, for each choice $k_{1},\ldots,k_{r}$, with
condition $k_{1}+\ldots+k_{r}=n$, $k_{i}\geq1$, $i=1,\ldots,r-1$ and
$k_{r}\geq2$,

there are $2^{k_{1}}\binom{n}{k_{1}}$ many ways to form $M_{1}$,

there are $2^{k_{2}}\binom{n-k_{1}}{k_{2}}$ many ways to form $M_{2}$,

$\ldots$,

there are $2^{k_{r}}\binom{n-k_{1}-\ldots-k_{r-1}}{k_{r}}$ many ways to form
$M_{r}$,

$b$ has two choices.

Hence,
\[
|\mathbb{NCF}(n,r)|=2\sum_{\substack{k_{1}+\ldots+k_{r}=n\\k_{i}%
\geq1,i=1,\ldots,r-1, k_{r}\geq2}}2^{k_{1}+\ldots+k_{r}}\binom{n}{k_{1}}%
\binom{n-k_{1}}{k_{2}}\ldots\binom{n-k_{1}-\ldots-k_{r-1}}{k_{r}}%
\]

\[
=2^{n+1}\sum_{\substack{k_{1}+\ldots+k_{r}=n\\k_{i}\geq1,i=1,\ldots,r-1,
k_{r}\geq2}}\frac{n!}{(k_{1})!(n-k_{1})!}\frac{(n-k_{1})!}{(k_{2}%
)!(n-k_{1}-k_{2})!}\ldots\frac{(n-k_{1}-\ldots-k_{r-1})!}{k_{r}!(n-k_{1}%
-\ldots-k_{r})!}%
\]

\[
=2^{n+1}\sum_{\substack{k_{1}+\ldots+k_{r}=n\\k_{i}\geq1,i=1,\ldots,r-1,
k_{r}\geq2}}\frac{n!}{k_{1}!k_{2}!\ldots k_{r}!}=2^{n+1}\sum_{\substack{k_{1}%
+\ldots+k_{r}=n\\k_{i}\geq1,i=1,\ldots,r-1, k_{r}\geq2}}\binom{n}{k_{1}%
,\ldots,k_{r-1}}.
\]

Since $\mathbb{NCF}(n)=\bigcup_{r=1}^{n-1}\mathbb{NCF}(n,r)$ and
$\mathbb{NCF}(n,i)\bigcap\mathbb{NCF}(n,j)=\phi$ when $i\neq j$, we get the
formula of $|\mathbb{NCF}(n)|$.
\end{proof}

One can check that $|\mathbb{NCF}(2)|=8$, $|\mathbb{NCF}(3)|=64$,
$|\mathbb{NCF}(4)|=736$, $|\mathbb{NCF}(5)|=10624$,...

These results are consistent with those in \cite{Ben, Sas}.

By equating our formula to the recursive relation in \cite{Ben, Sas}, we have
the following

\begin{cor}
\label{co3.2} The solution of the nonlinear recursive sequence%

\[
a_{2}=8, a_{n}=\sum_{r=2}^{n-1}\binom{n}{r-1}2^{r-1}a_{n-r+1}+2^{n+1} , n\geq3
\]

is
\[
a_{n}=2^{n+1}\sum_{\substack{r=1}}^{n-1}\sum_{\substack{k_{1}+\ldots
+k_{r}=n\\k_{i}\geq1,i=1,\ldots,r-1, k_{r}\geq2}}\binom{n}{k_{1}%
,\ldots,k_{r-1}}.%
\]

\end{cor}

\section{ Activity, Sensitivity and Hamming Weight}
A Boolean function is balanced if exactly half of its value is zero. Equivalently, the Hamming weight of this $n$ variables Boolean function is $2^{n-1}$. There are $\binom{2^n}{2^{n-1}}$ balanced  functions . It is easy to show that a Boolean functions with canalyzing variables is not balanced, i.e., biased. Actually, very biased. For example, Two constant functions are trivially canalyzing, They are the most biased. Extended monomial functions are the second most biased since for any of them, only one value is nonzero. But biased functions may have no canalyzing variables. For example, $f(x_1,x_2,x_3)=x_1x_2x_3\oplus x_1x_2\oplus x_1x_3\oplus x_2x_3$ is biased but without canalyzing variables.

In Boolean functions, some variable have greater influence over the output of the function than other variables. To formalize this, a  concept called $activity$ was introduced. Let $\frac{\partial f(x_1,\ldots,x_n)}{\partial x_i}=f(x_1,\ldots,x_i\oplus 1,\ldots,x_n)\oplus f(x_1,\ldots,x_i,\ldots,x_n)$. The $activity$ of variable $x_i$ of $f$ is defined as
\begin{equation}\label{act1}
\alpha_i^f=\frac{1}{2^n}\sum_{(x_1,\ldots,x_n)\in \mathbb{F}_2^n}\frac{\partial f(x_1,\ldots,x_n)}{\partial x_i}
\end{equation}
Note, the above definition can also be written as the following
\begin{equation}\label{act2}
\alpha_i^f=\frac{1}{2^{n-1}}\sum_{(x_1,\ldots,x_{i-1},x_{i+1},\ldots,x_n)\in \mathbb{F}_2^{n-1}}(f(x_1,\ldots,\overset{i}{0},\ldots,x_n)\oplus f(x_1,\ldots,\overset{i}{1},\ldots,x_n))
\end{equation}

The activity of any variables of constant functions is 0. For affine function $f(x_1,\ldots,x_n)=x_1\oplus \ldots\oplus x_n\oplus b$, $\alpha_i^f=1$ for any $i$. It is clear, for any $f$ and $i$, we have $0\leq \alpha_i^f\leq 1$.

Another important quantity is the sensitivity of a Boolean function, which measures how sensitive the output of the function is if the input changes (This was introduced in \cite {Coo}). The sensitivity $s^f(x_1,\ldots,x_n)$ of $f$ on vector $(x_1,\ldots,x_n)$ is defined as
 the number of Hamming neighbors of $(x_1,\ldots,x_n)$ on which the function value is different from $f(x_1,\ldots,x_n)$. That is,
\begin{equation*}
s^f(x_1,\ldots,x_n)=|\{i|f(x_1,\ldots,\overset{i}{0},\ldots,x_n)\neq f(x_1,\ldots,\overset{i}{1},\ldots,x_n), i=1,\ldots,n \}|.
\end{equation*}

Obviously, $s^f(x_1,\ldots,x_n)=\sum_{i=1}^n\frac{\partial f(x_1,\ldots,x_n)}{\partial x_i}$

The average sensitivity of function $f$ is defined as

\begin{equation*}
s^f=E[s^f(x_1,\ldots,x_n)]=\frac{1}{2^n}\sum_{(x_1,\ldots,x_n)\in \mathbb{F}_2^n}s^f(x_1,\ldots,x_n)=\sum_{i=1}^n\alpha_i^f.
\end{equation*}
It is clear that $0\leq s^f\leq n$.

The average sensitivity is one of the most studied concepts in the analysis of Boolean functions. Recently, It receives a lot of attention.
 See \cite{Ama, Ber, Ber2, Bop, Che, Chr, Kel, Liu, Li, Qia, Shm2, Shm, Shp,  Sch, Vir}. Bernasconi \cite{Ber} has showed that a random Boolean function has average sensitivity $\frac{n}{2}$. It means the average value of the average sensitivities of all the $n$ variables Boolean functions is $\frac{n}{2}$. In \cite {Shm}, Ilya Shmulevich and Stuart A. Kauffman calculated the activity of all the variables of a Boolean functions with exactly one canalyzing variable and unbiased input for the other variable. Add all the activities, the average sensitivity of this function was also obtained.

In the following, using Equation \ref{eq3.1}, we will obtain the formula of the Hamming weight of any NCF, the  activities of all the variables of any NCF and the average sensitivity (which is bounded by constant) of any NCF.

First, we have
\begin{lemma}\label{lm4.1}
 $(x_1\oplus a_1)\ldots (x_k\oplus a_k)=$
$\left\{
\begin{array}{ll}
1, & (x_1,\ldots,x_k)=(\overline{a_1},\ldots,\overline{a_k})\\
0, & otherwise.
\end{array}\right.$
i.e., only one value is $1$ and all the other $2^k-1$ values are $0$.
\end{lemma}

\begin{theorem}\label{th4.1}
Given $n\geq 2$. Let $f_1=M_1$, $f_r=M_{1}(M_{2}(\ldots(M_{r-1}%
(M_{r}\oplus 1)\oplus 1)\ldots)\oplus 1)$ , $r\geq 2$, where $M_i$ is same as that in the in Theorem \ref{th3.2},  Then the Hamming weight of $f_r$ is
\begin{equation}\label{eq4.3}
 W(f_r)=\sum_{j=1}^r(-1)^{j-1}2^{n-\sum_{i=1}^jk_i}
\end{equation}
The Hamming weight of $f_r\oplus 1$ is
\begin{equation}\label{eq4.4}
 W(f_r\oplus 1)=\sum_{j=0}^r(-1)^{j}2^{n-\sum_{i=1}^jk_i}
\end{equation}
Where $\sum_{i=1}^0k_i$ should be explained as $0$.
\end{theorem}

\begin{proof}

First, let's consider the Hamming weight of $f_r$.

When $r=1$, we know the result is true by Lemma \ref{lm4.1}.

When $r>1$, we consider two cases:

Case A: $r$ is odd, $r=2t+1$.
 
 All the vectors make $f=1$ will be divided into the following disjoint groups.

Group $1$: $M_1=1$, $M_2=0$;

Group $2$: $M_1=1$, $M_2=1$, $M_3=1$, $M_4=0$;

$\ldots$

Group $j$: $M_1=1$, $M_2=1$, $\ldots$, $ M_{2j-1}=1$, $M_{2j}=0$;

$\ldots$

Group $t$ : $M_1=1$, $M_2=1$, $\ldots$, $M_{2t-1}=1$, $M_{2t}=0$;

Group $t+1$ : $M_1=1$, $M_2=1$, $\ldots$, $M_{2t}=1$, $M_{2t+1}=M_r=1$.

In Group $1$, the number of  vectors is $(2^{k_2}-1)2^{n-k_1-k_2}=2^{n-k_1}-2^{n-k_1-k_2}$.

In Group $2$, the number of vector is $(2^{k_4}-1)2^{n-k_1-k_2-k_3-k_4}=2^{n-k_1-k_2-k_3}-2^{n-k_1-k_2-k_3-k_4}$.

\ldots

In Group $t$, the number of vector is $(2^{k_{2t}}-1)2^{n-k_1-\ldots-k_{2t}}=2^{n-k_1-\ldots-k_{2t-1}}-2^{n-k_1-\ldots -k_{2t}}$.

In Group $t+1$, the number of vectors is $2^{n-k_1-\ldots -k_r}=1$.

Add all of them, we get the formula Equation \ref{eq4.3}.

Case B: $r$ is even, $r=2t$.
 
 All the vectors make $f=1$ will be divided into the following disjoint groups.

Group $1$: $M_1=1$, $M_2=0$;

Group $2$: $M_1=1$, $M_2=1$, $M_3=1$, $M_4=0$;

$\ldots$

Group $j$: $M_1=1$, $M_2=1$, $\ldots$, $ M_{2j-1}=1$, $M_{2j}=0$;

$\ldots$

Group $t-1$ : $M_1=1$, $M_2=1$, $\ldots$, $M_{2t-3}=1$, $M_{2t-2}=0$;

Group $t$ : $M_1=1$, $M_2=1$, $\ldots$, $M_{2t-1}=1$, $M_{2t}=M_r=0$.

In Group $1$, the number of  vectors is $(2^{k_2}-1)2^{n-k_1-k_2}=2^{n-k_1}-2^{n-k_1-k_2}$.

In Group $2$, the number of vector is $(2^{k_4}-1)2^{n-k_1-k_2-k_3-k_4}=2^{n-k_1-k_2-k_3}-2^{n-k_1-k_2-k_3-k_4}$.

\ldots

In Group $t-1$, the number is $(2^{k_{2t-2}}-1)2^{n-k_1-\ldots-k_{2t-2}}=2^{n-k_1-\ldots-k_{2t-3}}-2^{n-k_1-\ldots -k_{2t-2}}$.

In Group $t$, the number of vectors is $2^{n-k_1-\ldots -k_{2t-1}}-2^{n-k_1-\ldots -k_{2t}}=2^{k_{2t}}-1$.

Add all of them, we get the formula Equation \ref{eq4.3} again.

Because  $|\{(x_1,\ldots,x_n)|f(x_1,\ldots,x_n)=0\}|+|\{(x_1,\ldots,x_n)|f(x_1,\ldots,x_n)=1\}|=2^n$,
 we know the Hamming weight of $f_r\oplus 1$ is 
 \begin{equation*}
 W(f_r\oplus 1)=2^n-W(f_r)=2^n-\sum_{j=1}^r(-1)^{j-1}2^{n-\sum_{i=1}^jk_i}=\sum_{j=0}^r(-1)^{j}2^{n-\sum_{i=1}^jk_i}.
 \end{equation*}
Where $\sum_{i=1}^0k_i$ should be explained as $0$.
\end{proof}

In the following, we will calculate the activities of the variables of any NCF.

Let $f$ be a NCF and written as the form in Theorem \ref{th3.2}. Without loss of generality(to avoid the complicated notation), we assume
 $M_1=(x_1\oplus a_1)(x_2\oplus a_2)\ldots (x_{k_1}\oplus a_{k_1})$ and $m_1=(x_1\oplus a_1)\ldots(x_{i-1}\oplus a_{i-1})(x_{i+1}\oplus a_{i+1})\ldots (x_{k_1}\oplus a_{k_1})$, i.e., $M_1=(x_i\oplus a_i)m_1$.
 
 If $r=1$, i.e., $k_1=n$, then 
 
 \begin{equation*}
\alpha_i^f=\frac{1}{2^{n-1}}\sum_{(x_1,\ldots,x_{i-1},x_{i+1},\ldots,x_n)\in \mathbb{F}_2^{n-1}}(f(x_1,\ldots,\overset{i}{0},\ldots,x_n)\oplus f(x_1,\ldots,\overset{i}{1},\ldots,x_n))
\end{equation*}

\begin{equation*}
=\frac{1}{2^{n-1}}\sum_{(x_1,\ldots,x_{i-1},x_{i+1},\ldots,x_n)\in \mathbb{F}_2^{n-1}}m_1=\frac{1}{2^{n-1}}W(m_1)=\frac{1}{2^{n-1}}.
\end{equation*}
by Lemma \ref{lm4.1}.

If $1<r\leq n-1$,
 
 Let's consider the activity of $x_i$ in the first layer, i.e.,  $1\leq i\leq k_1$. We have
 
\begin{equation*}
\alpha_i^f=\frac{1}{2^{n-1}}\sum_{(x_1,\ldots,x_{i-1},x_{i+1},\ldots,x_n)\in \mathbb{F}_2^{n-1}}(f(x_1,\ldots,\overset{i}{0},\ldots,x_n)\oplus f(x_1,\ldots,\overset{i}{1},\ldots,x_n))
\end{equation*}

\begin{equation*}
=\frac{1}{2^{n-1}}\sum_{(x_1,\ldots,x_{i-1},x_{i+1},\ldots,x_n)\in \mathbb{F}_2^{n-1}}m_1(M_{2}(\ldots(M_{r-1}%
(M_{r}\oplus 1)\oplus 1)\ldots)\oplus 1).
\end{equation*}

\begin{equation*}
=\frac{1}{2^{n-1}}W(m_1(M_{2}(\ldots(M_{r-1}%
(M_{r}\oplus 1)\oplus 1)\ldots)\oplus 1)).
\end{equation*}
$=\left\{
\begin{array}{ll}
\frac{1}{2^{n-1}}\sum_{j=1}^r(-1)^{j-1}2^{n-1-(\sum_{i=1}^jk_i-1)}, & k_1>1\\
\frac{1}{2^{n-1}}\sum_{j=0}^{r-1}(-1)^{j}2^{n-1-\sum_{i=1}^jk_{i+1}}, & k_1=1.
\end{array}\right.$

$=\left\{
\begin{array}{ll}
\frac{1}{2^{n-1}}\sum_{j=1}^r(-1)^{j-1}2^{n-\sum_{i=1}^jk_i}, & k_1>1\\
\frac{1}{2^{n-1}}\sum_{j=0}^{r-1}(-1)^{j}2^{n-\sum_{i=0}^jk_{i+1}}, & k_1=1.
\end{array}\right.=\left\{
\begin{array}{ll}
\frac{1}{2^{n-1}}\sum_{j=1}^r(-1)^{j-1}2^{n-\sum_{i=1}^jk_i}, & k_1>1\\
\frac{1}{2^{n-1}}\sum_{j=0}^{r-1}(-1)^{j}2^{n-\sum_{i=1}^{j+1}k_{i}}, & k_1=1.
\end{array}\right.$

$=\left\{
\begin{array}{ll}
\frac{1}{2^{n-1}}\sum_{j=1}^r(-1)^{j-1}2^{n-\sum_{i=1}^jk_i}, & k_1>1\\
\frac{1}{2^{n-1}}\sum_{j=1}^r(-1)^{j-1}2^{n-\sum_{i=1}^jk_i}, & k_1=1.
\end{array}\right.=\frac{1}{2^{n-1}}\sum_{j=1}^r(-1)^{j-1}2^{n-\sum_{i=1}^jk_i}$

by Theorem \ref{th4.1}. Note, in the above, $k_1=1$ means $m_1=1$,  we used the Equation \ref{eq4.4} with layer number $r-1$ and the first layer is $M_2$ for $n-1$ variables functions.

Now let's consider the variables in the second layer, i.e., $x_i$ is an essential variable of $M_2$. We have
$M_2=(x_i+a_i)m_2$ and 

\begin{equation*}
\alpha_i^f=\frac{1}{2^{n-1}}\sum_{(x_1,\ldots,x_{i-1},x_{i+1},\ldots,x_n)\in \mathbb{F}_2^{n-1}}(f(x_1,\ldots,\overset{i}{0},\ldots,x_n)\oplus f(x_1,\ldots,\overset{i}{1},\ldots,x_n))
\end{equation*}

\begin{equation*}
=\frac{1}{2^{n-1}}\sum_{(x_1,\ldots,x_{i-1},x_{i+1},\ldots,x_n)\in \mathbb{F}_2^{n-1}}M_1(m_{2}(\ldots(M_{r-1}%
(M_{r}\oplus 1)\oplus 1)\ldots)).
\end{equation*}
\begin{equation*}
=\frac{1}{2^{n-1}}\sum_{(x_1,\ldots,x_{i-1},x_{i+1},\ldots,x_n)\in \mathbb{F}_2^{n-1}}M_1m_{2}(\ldots(M_{r-1}%
(M_{r}\oplus 1)\oplus 1)\ldots).
\end{equation*}
\begin{equation*}
=\frac{1}{2^{n-1}}\sum_{j=1}^{r-1}(-1)^{j-1}2^{n-1-((k_1+k_2-1)+\ldots +k_{j+1}))}=\frac{1}{2^{n-1}}\sum_{j=1}^{r-1}(-1)^{j-1}2^{n-\sum_{i=1}^{j+1}k_i}
\end{equation*}

by Equation \ref{eq4.3} in Theorem \ref{th4.1}. Note, $M_1m_2$ is the first layer, $M_3$ is the second layer and etc.

Now let's consider the variables in the $lth$  layer, i.e., $x_i$ is an essential variable of $M_l$, $2\leq l\leq r-1$. We have
$M_l=(x_i+a_i)m_l$ and 

\begin{equation*}
\alpha_i^f=\frac{1}{2^{n-1}}\sum_{(x_1,\ldots,x_{i-1},x_{i+1},\ldots,x_n)\in \mathbb{F}_2^{n-1}}(f(x_1,\ldots,\overset{i}{0},\ldots,x_n)\oplus f(x_1,\ldots,\overset{i}{1},\ldots,x_n))
\end{equation*}

\begin{equation*}
=\frac{1}{2^{n-1}}\sum_{(x_1,\ldots,x_{i-1},x_{i+1},\ldots,x_n)\in \mathbb{F}_2^{n-1}}M_1\ldots M_{l-1}m_l(M_{l+1}(\ldots (M_r\oplus 1)\ldots )\oplus 1).
\end{equation*}

\begin{equation*}
=\frac{1}{2^{n-1}}\sum_{j=1}^{r-l+1}(-1)^{j-1}2^{n-1-((k_1+\ldots + k_{l}-1)+k_{l+1}+\ldots  +k_{j+l-1}))}=\frac{1}{2^{n-1}}\sum_{j=1}^{r-l+1}(-1)^{j-1}2^{n-\sum_{i=1}^{j+l-1}k_i}
\end{equation*}

by Equation \ref{eq4.3} in Theorem \ref{th4.1}. Note, $M_1\ldots M_{l-1}m_l$ is the first layer, $M_{l+1}$ is the second layer, and etc.

Let $x_i$ be the variable in the last layer $M_{r}$, we have

\begin{equation*}
=\frac{1}{2^{n-1}}\sum_{(x_1,\ldots,x_{i-1},x_{i+1},\ldots,x_n)\in \mathbb{F}_2^{n-1}}M_1M_2\ldots\ M_{r-1}m_r=\frac{1}{2^{n-1}}
\end{equation*} by Lemma \ref{lm4.1}.

Variables in the same layer have the same activities, so we use $A_l^f$ to stand for the activity number of each variable in the $lth$ layer $M_l$, $1\leq l\leq r$. We find the formula of $A_l^f$ for $2\leq l\leq r-1$ is also true when $l=r$ or $r=1$. Hence, we write all the above as the following

\begin{theorem}\label{th4.2}
Let $f$ be a NCF and written as in the Theorem \ref{th3.2}.

then the activity of each variable in the $lth$  layer , $1\leq l\leq r$, is

\begin{equation}\label{4.5}
A_l^f=\frac{1}{2^{n-1}}\sum_{j=1}^{r-l+1}(-1)^{j-1}2^{n-\sum_{i=1}^{j+l-1}k_i}
\end{equation}

The average sensitivity of $f$ is 

\begin{equation}\label{eq4.6}
s^f=\sum_{l=1}^rk_lA_l^f=\frac{1}{2^{n-1}}\sum_{l=1}^r k_l\sum_{j=1}^{r-l+1}(-1)^{j-1}2^{n-\sum_{i=1}^{j+l-1}k_i}
\end{equation}
\end{theorem}
 We do some analysis about the formulas in Theorem \ref{th4.2}, we have

 \begin{cor}\label{co4.1}
 $n\geq 3$, $A_1^f>A_2^f>\ldots >A_r^f$ and $\frac{n}{2^{n-1}}\leq  s^f < 2- \frac{1}{2^{n-2}}$ 
 \end{cor}

\begin{proof}
\begin{equation*}
A_l^f=\frac{1}{2^{n-1}}\sum_{j=1}^{r-l+1}(-1)^{j-1}2^{n-\sum_{i=1}^{j+l-1}k_i}=\frac{1}{2^{n-1}}(2^{n-k_1-\ldots -k_l}-2^{n-k_1-\ldots -k_{l+1}}+\ldots (-1)^{r-l})
\end{equation*}
Since the sum is an alternate decreasing  sequence and $k_{l+1}\geq 1$, we have
\begin{equation*}
\frac{1}{2^{n-1}}(2^{n-k_1-\ldots -k_l-1})\leq \frac{1}{2^{n-1}}(2^{n-k_1-\ldots -k_l}-2^{n-k_1-\ldots -k_{l+1}})< A_l^f< \frac{1}{2^{n-1}}(2^{n-k_1-\ldots -k_l})
\end{equation*}
Hence,
\begin{equation*}
 A_{l+1}^f< \frac{1}{2^{n-1}}(2^{n-k_1-\ldots -k_{l+1}})\leq \frac{1}{2^{n-1}}(2^{n-k_1-\ldots -k_l-1})<A_l^f.
\end{equation*}

We have 

\begin{equation*}
k_1A_1^f=\frac{k_1}{2^{n-1}}(2^{n-{k_1}}-2^{n-{k_1}-k_2}+2^{n-{k_1}-k_2-k_3}-\ldots (-1)^{r-1})
\end{equation*}

\begin{equation*}
k_2A_2^f=\frac{k_2}{2^{n-1}}(2^{n-k_1-k_2}-2^{n-k_1-k_2-k_3}+2^{n-k_1-k_2-k_3-k_4}-\ldots (-1)^{r-2})
\end{equation*}

\begin{equation*}
\ldots \ldots
\end{equation*}

\begin{equation*}
k_lA_l^f=\frac{k_l}{2^{n-1}}(2^{n-k_1-\ldots-k_l}-2^{n-k_1-\ldots-k_l-k_{l+1}}-\ldots (-1)^{r-l})
\end{equation*}
\begin{equation*}
\ldots \ldots
\end{equation*}

\begin{equation*}
k_rA_r^f=\frac{k_r}{2^{n-1}}
\end{equation*}

Hence, $s^f=\sum_{l=1}^rk_lA_l^f\geq \frac{k_1}{2^{n-1}}+\frac{k_2}{2^{n-1}}+\ldots +\frac{k_r}{2^{n-1}}=\frac{n}{2^{n-1}}$, so we know the NCF with LAYER NUMBER 1 has the minimal average sensitivity.

On the other hand, $s^f=\sum_{l=1}^rk_lA_l^f<\frac{k_1}{2^{n-1}}2^{n-k_1}+\frac{k_2}{2^{n-1}}2^{n-k_1-k_2}+\ldots +\frac{k_l}{2^{n-1}}2^{n-k_1-\ldots -k_l}+\dots+\frac{k_r}{2^{n-1}}=U(k_1,\ldots,k_r)$, where $k_1+\ldots +k_r=n$, $k_i\geq 1$, $i=1,\ldots, r-1$ and $k_r\geq 2$. We will find the maximal value of $U(k_1,\ldots,k_r)$ in the following.

First, we claim $k_r=2$ if $U(k_1,\ldots,k_r)$ reach maximal value.  Because if $k_r$ is increased by $1$, and the last term makes $\frac{1}{2^{n-1}}$ more contributions to $U(k_1,\ldots,k_r)$, then there exists $l$, $k_l$ will be decreased by $1$ ($k_1+\ldots+k_r=n$), hence 

\begin{equation*}
\frac{k_l}{2^{n-1}}2^{n-k_1-\ldots -k_l}
\end{equation*}

 will be decreased more than $\frac{1}{2^{n-1}}$.

Now, Look at $\frac{k_1}{2^{n-1}}2^{n-k_1}$, it is obvious it attains the maximal value only when $k_1=1$ or $2$ but obviously $k_1=1$ will be the choice since it also make all the other terms greater.. 

Now Look at $\frac{k_2}{2^{n-1}}2^{n-k_1-k_2}$, it attains the maximal value  when $k_1=k_2=1$ or $k_1=1$ and $k_2=2$, again, $k_2=1$ is the best choice  to make all the other terms greater.
 
In general, if $k_1=\ldots=k_{l-1}=1$, then $\frac{k_l}{2^{n-1}}2^{n-k_1-\ldots -k_l}$  attains its maximal value when $k_l=1$, where $1\leq l\leq r-1$.

In summary, we have showed that $U(k_1,\ldots,k_r)$ reaches maximal value when $r=n-1$, $k_1=\ldots=k_{n-2}=1$, $k_{n-1}=2$ and 

Max $ U(k_1,\ldots,k_r)=U(1,\ldots,1,2)=\frac{1}{2^{n-1}}(2^{n-1}+2^{n-2}+\ldots+ 2^{2}+2)=2-\frac{1}{2^{n-2}}$.
\end {proof}

\begin{remark}
So, we know the average sensitivity is bounded by constants for any NCF with any number of variables Since the minimal value approaches to $0$ and the maximal value of $U(k_1,\ldots,k_r)$ approaches to $2$ as $n\rightarrow \infty$. Hence, $0<s^f<2$ for any NCF with arbitrary number of variables. 
\end{remark}

In the following, we evaluate the formula Equation \ref{eq4.6} for some parameters $k_1,\ldots,k_r$, we have
\begin{lemma}\label{lm4.2}
1) When $r=n-1$, $k_1=\ldots=k_{n-2}=1$, $k_{n-1}=2$, $s^f=\frac{4}{3}-\frac{3+(-1)^n}{3\times 2^n}$;

2) Given $n\geq 4$, $r=n-2$, $k_1=\ldots=k_{n-3}=1$, $k_{n-2}=3$, $s^f=\frac{4}{3}-\frac{9+5(-1)^{n-1}}{3\times 2^n}$;

3) If $n$ is even and $n\geq 6 $, $r=\frac{n}{2}$, $k_1=1$, $k_2=\ldots=k_{\frac{n}{2}-1}=2$, $k_{\frac{n}{2}}=3$, $s^f=\frac{4}{3}-\frac{4}{3\times 2^n}$. Hence, these three cardinalities are equal if $n$ is even.
\end{lemma}

\begin{proof}
When $r=n-1$, $k_1=\ldots=k_{n-2}=1$, $k_{n-1}=2$ by Equation \ref{eq4.6}. We have

\begin{equation*}
s^f=\sum_{l=1}^{r}k_lA_l^f=\frac{1}{2^{n-1}}\sum_{l=1}^{n-1} k_l\sum_{j=1}^{n-l}(-1)^{j-1}2^{n-\sum_{i=1}^{j+l-1}k_i}
\end{equation*}

\begin{equation*}
=\frac{1}{2^{n-1}}\sum_{l=1}^{n-1} k_l(\sum_{j=1}^{n-l-1}(-1)^{j-1}2^{n-j-l+1}+(-1)^{n-l-1})=\frac{1}{2^{n-1}}\sum_{l=1}^{n-1} k_l(\frac{1}{3}2^{n-l+1}+\frac{1}{3}(-1)^{n-l})
\end{equation*}

\begin{equation*}
=\frac{1}{2^{n-1}}(\sum_{l=1}^{n-2} (\frac{1}{3}2^{n-l+1}+\frac{1}{3}(-1)^{n-l})+2)=\frac{4}{3}-\frac{3+(-1)^n}{3\times 2^n}
\end{equation*}
 The other two formulas are also routine simplifications of Equation \ref{eq4.6}.
\end{proof}

Based on our numerical calculation, Lemma \ref{lm4.2} and the proof of Corollary \ref{co4.1}, We have the following

\begin{conj}\label{conj4.1}
The maximal value of $s^f$ is $s^f=\frac{4}{3}-\frac{3+(-1)^n}{3\times 2^n}$. It will be reached if the NCF has the maximal LAYER NUMBERS $n-1$, i.e., if $r=n-1$, $k_1=\ldots=k_{n-2}=1$, $k_{n-1}=2$. When $n$ is even, this maximal value is also reached by NCF with parameters 
$n\geq 4$, $r=n-2$, $k_1=\ldots=k_{n-3}=1$, $k_{n-2}=3$ or $n\geq 6 $, $r=\frac{n}{2}$, $k_1=1$, $k_2=\ldots=k_{\frac{n}{2}-1}=2$ and $k_{\frac{n}{2}}=3$.
\end{conj}

\begin{remark}\label{re4.2}
When $n=6$, the NCF with $k_1=1$, $k_2=2$, $k_3=1$ and $k_4=2$ also has the maximal average sensitivity $\frac{21}{16}$. But this can not be generalized. If the above conjecture is true, then we have $0<s^f<\frac{4}{3}$ for any NCF with arbitrary number of variables. In other words, both $0$ and $\frac{4}{3}$ are uniform tight bounds for any NCF.
\end{remark}

We point out, given the algebraic normal form of $f$, it is easy to find all of its canalyzing variables (the first layer $M_1$), then 
write $f=M_1g+b$, repeating the schedule to $g$, we can easily to determine if $f$ is NCF, if yes, we then write it as the form in Theorem \ref{th3.2}.

We end this section by the following example.
\begin{example}
Let $N(x_1,x_2,x_3,x_4)=x_1x_2x_3\oplus x_2x_3x_4\oplus x_1x_3\oplus x_3x_4\oplus 1$ and 

$Y(x_1,x_2,x_3,x_4,x_5)=x_1x_2x_3x_4x_5\oplus x_1x_2x_3x_4\oplus x_1x_2x_4x_5\oplus x_1x_2x_4\oplus x_1x_3x_4\oplus x_1x_3\oplus x_1x_4\oplus x_1$.

For $N(x_1,x_2,x_3,x_4)$, for all the $4$ variables, we found when $x_2=1$ or $x_3=0$, then functions becomes constant $1$, so we know 
$N(x_1,x_2,x_3,x_4)=(x_2\oplus 1)(x_3)N_1\oplus 1$. Actually, 

$N(x_1,x_2,x_3,x_4)=x_1x_2x_3\oplus x_2x_3x_4\oplus x_1x_3\oplus x_3x_4\oplus 1=x_3(x_1x_2\oplus x_2x_4\oplus x_1\oplus x_4)\oplus 1$

$=x_3(x_2(x_1\oplus x_4)\oplus x_1\oplus x_4)\oplus 1=x_3((x_2\oplus 1)(x_1\oplus x_4))\oplus 1$. Since $x_1\oplus x_4$ has no canalyzing variable, we know $N$ is not NCF, but a partially NCF.

For $Y(x_1,x_2,x_3,x_4,x_5)$, We find $x_1=0$ or $x_3=1$, the function will be reduced to $0$, so we know $Y=x_1(x_3\oplus 1)Y_1$. 

Where $Y_1=x_2x_4x_5\oplus x_2x_4\oplus x_4\oplus 1$, for this function we find only when $x_4=0$, $Y_1$ will be reduced to $1$, so $Y_1=x_4Y_2\oplus 1$, where $Y_2=x_2x_5\oplus x_2\oplus 1$, and finally, we have $Y_2=x_2(x_5\oplus 1)\oplus 1$, So $Y$ is NCF with $n=5$, $r=3$ and $k_1=2$, $k_2=1$, $k_3=2$, $M_1=x_1(x_3\oplus 1)$, $M_2=x_4$ and $M_3=x_2(x_5\oplus 1)$, hence its Hamming weight is $5$ by Equation \ref{eq4.3} and its average sensitivity is $\frac{15}{16}$by Equation \ref{eq4.6}.

\end{example}

\section{Conclusion}

We obtain a complete characterization for nested canalyzing functions
(NCFs) by deriving its unique algebraic normal form (polynomial form). We
introduced a new invariant, LAYER NUMBER for nested canalyzing function. So,
the dominance of nested canalyzing variables is quantified. Consequently, we
obtain the explicit formula of the number of nested canalyzing functions. Based on the
polynomial form, we also obtain the formula of the Hamming weight of each NCF. 
The activity number of each variable of a NCF is also provided with an explicit formula. Consequently, we proved the average sensitivity
of any NCF is less than $2$, hence, we proved why NCF is stable theoretically. Finally, we conjecture that the tight upper bound for the average sensitivity of any NCF is $\frac{4}{3}$.


\begin{thebibliography}{99}   

\bibitem {Ama}Kazuyuki Amano,
 \newblock\textquotedblleft
Tight bounds on the average sensitivity of k-CNF,
\textquotedblright Theory of Computing, Vol 7 (2011), pp. 45-48.



\bibitem {Bal}E. Balleza, E. R. Alvarez-Buylla, A. Chaos, S. Kauffman, I.
Shmulevich,and M. Aldana, \newblock\textquotedblleft Critical dynamics in
genetic regulatory networks: Examples from four kingdoms, \textquotedblright
PLoS ONE, 3 (2008), p. e2456

\bibitem {Bar}C. Barrett, C. Herring, J. Reed, and B. Palsson,
\newblock\textquotedblleft The global transcriptional regulatory network for
metabolism in Escherichia coli exhibits few dominant functional states,
\textquotedblright Proc Natl Acad Sci USA, 102 (2005), pp. 19103–19108.
                                                                                            %

\bibitem {Ben}E. A. Bender, J. T. Butler, \newblock\textquotedblleft
Asymptotic approximations for the number of fanout-tree functions,
\textquotedblright IEEE Trans. Comput. 27 (12) (1978) 1180-1183.

\bibitem {Ber}A. Bernasconi,
 \newblock\textquotedblleft
Mathematical techniques for the analysis of Boolean functions,
\textquotedblright Ph.D. thesis, Dipartmento di Informatica, Universita di Pisa (March, 1998).

\bibitem {Ber2}A. Bernasconi,
 \newblock\textquotedblleft
Sensitivity vs. block sensitivity (an average-case study)
\textquotedblright Information processing letters 59 (1996) 151-157.

\bibitem {Bop}Ravi B. Boppana,
 \newblock\textquotedblleft
The average sensitivity of bounded-depth circuits
\textquotedblright Information processing letters 63 (1997) 257-261.




\bibitem {But}J. T. Butler, T. Sasao, and M. Matsuura,
\newblock\textquotedblleft Average path length of binary decision diagrams,
\textquotedblright IEEE Transactions on Computers, 54 (2005), pp. 1041–1053.


\bibitem {Can}David Canright, Sugata Gangopadhyay, Subhamoy Maitra, Pantelimon St$\breve{a}$nic$\breve{a}$,
\newblock\textquotedblleft Laced Boolean functions and subset sum problems in finite fields,
\textquotedblright Discrete applied mathematics, 159 (2011), pp. 1059-1069.

\bibitem {Che}Shijian Chen and Yiguang Hong,
 \newblock\textquotedblleft
Control of random Boolean networks via average sensitivity of Boolean functions,
\textquotedblright Chin. Phys. B Vol. 20, No 3 (2011) 036401.
\bibitem {Chr}Demetres Christofides,
 \newblock\textquotedblleft
Influences of Monotone Boolean Functions,
\textquotedblright Preprint 2009.

\bibitem {Coo}S. A. Cook, C. Dwork, R. Reischuk,
\newblock\textquotedblleft Upper and lower time bounds for parallel random access machines without simultaneous writes,
\textquotedblright SIAM J. Comput, 15 (1986), pp. 87-89.



\bibitem {Har}S. E. Harris, B. K. Sawhill, A. Wuensche, and S. Kauffman,
\newblock\textquotedblleft A model of transcriptional regulatory networks
based on biases in the observed regulation rules, \textquotedblright Complex,
7 (2002), pp. 23–40.

\bibitem {Herr}M. Herrgard, B. Lee, V. Portnoy, and B. Palsson,
\newblock\textquotedblleft Integrated analysis of regulatory and metabolic
networks reveals novel regulatory mechanisms in saccharomyces cerevisiae,
\textquotedblright Genome Res, 16 (2006), pp. 627–635.



\bibitem {Jar}A. Jarrah, R. Laubenbacher, and A. Veliz-Cuba,
\newblock\textquotedblleft A polynomial framework for modeling and anaylzing
logical models. \textquotedblright In Preparation, 2008.

\bibitem {Abd2}A. Jarrah, B. Ropasa and R. Laubenbacher,
\newblock\textquotedblleft Nested Canalyzing, Unate Cascade, and Polynomial
Functions\textquotedblright, \newblock {\em Physica D} 233 (2007), pp. 167-174.

\bibitem {Win}Winfried Just, \newblock\textquotedblleft The steady state
system problem is NP-hard even for monotone quadratic Boolean dynamical
systems \textquotedblright Preprint,2006

\bibitem {Win2}Winfried Just, Ilya Shmulevich, John Konvalina, \newblock ``The
number and probability of canalyzing functions'', \newblock {\em Physica D}
197 (2004), pp. 211-221.

\bibitem {Kau1}S. A. Kauffman, \newblock\textquotedblleft The Origins of
Order: Self-Organization and Selection in Evolution\textquotedblright,
\newblock {\em Oxford University Press, New York, Oxford} (1993).

\bibitem {Kau2}S. A. Kauffman, C. Peterson, B. samuelesson, C. Troein,
\newblock\textquotedblleft Random Boolean Network Models and the Yeast
Transcription Network\textquotedblright, \newblock {\em Proc. Natl. Acad. Sci}
100 (25) (2003), pp. 14796-14799.

\bibitem {Kau3}S. A. Kauffman, C. Peterson, B. Samuelsson, and C. Troein,
\newblock\textquotedblleft Genetic networks with canalyzing Boolean rules are
always stable, \textquotedblright, \newblock {\em PNAS}, 101 (2004), pp. 17102–17107.

\bibitem {Kel}N.Keller and H. Pilpel,
 \newblock\textquotedblleft
Linear transformations on monotone functions on the discrete cube,
\textquotedblright Discrete Math. 309 (2009), 4210-4214.

\bibitem {Lau}R. Laubenbacher and B. Pareigis,
\newblock\textquotedblleft
Equivenlence relations on finite dynamical systems, \textquotedblright,
\newblock {\em Advances in applied mathematics} 26 (2001), pp. 237-251.

\bibitem {Liu}W. Liu, H. L$\ddot{a}$hdedm$\ddot{a}$ki, Edward R. Dougherty and I. Shmulevich,
\newblock\textquotedblleft
Inference of Boolean Networks Using Sensitivity Regularization, 
\textquotedblright,
\newblock {\em EURASIP Journal of Bioinformatics and System Biology} Volume 2008, Article ID 780541, 12 pages.


\bibitem {Lor}Lori Layne, \newblock ``Biologically Relevant Classes of Boolean
Functions'', \newblock {\em Ph.D Thesis, Clemson University} (2011).


\bibitem {Li}Jiyou Li,
 \newblock\textquotedblleft
On the average sensitivity of the weighted sum function,
\textquotedblright arXiv:1108.3198v2 [cs.IT] 18 Aug 2011.



\bibitem {Yua}Yuan Li, \newblock\textquotedblleft Results on Rotation
Symmetric Polynomials Over $GF(p)$\textquotedblright,
\newblock {\em Information Sciences} 178 (2008), pp. 280-286.

\bibitem {Yua2}Yuan Li, David Murragarra, John O Adeyeye and Reinhard
Laubenbacher \newblock\textquotedblleft Multi-State Canalyzing Functions over
Finite Fields \textquotedblright, Preprint (2010)


\bibitem {Lid}R. Lidl and H. Niederreiter, \newblock ``Finite Fields'',
\newblock {\em Cambridge University Press, New York} (1977).



\bibitem {Mor}A. A. Moreira and L. A. Amaral, \newblock\textquotedblleft
Canalyzing Kauffman networks: Nonergodicity and its effect on their critical
behavior, \textquotedblright, \newblock {\em Phys. Rev. lett.} 94 (21) (2005), 218702.

\bibitem {Mur}D. Murragarra and R. Laubenbacher, \newblock\textquotedblleft
Generalized nested canalyzing functions capture the logic of gene regulation,
\textquotedblright(2010), p. under review.

\bibitem {Mur2}D. Murragarra and R. Laubenbacher, \newblock\textquotedblleft
The number of multistate nested canalyzing functions, \textquotedblright arXiv
1108.0206v2 [math.AG] 7 Aug 2011.

\bibitem {Nik}S. Nikolajewaa, M. Friedela, and T. Wilhelm,
\newblock\textquotedblleft Boolean networks with biologically relevant rules
show ordered behaviorstar, open, \textquotedblright Biosystems, 90 (2007), pp. 40–47.

\bibitem {Nis}N. Nisan,
\newblock\textquotedblleft CREW PRAMs and decision tree,
 \textquotedblright SIAM J. Comput, 20 (6) (1991), PP. 999-1070.
 
 \bibitem {Qia}Xiaoning Qian and Edward R. Dougherty,
 \newblock\textquotedblleft
A comparative study on sensitivitys of Boolean networks,
\textquotedblright 978-1-61284-792-4/10 2011 IEEE.


\bibitem {Sas}T. Sasao, K. Kinoshita, \newblock\textquotedblleft On the number
of fanout-tree functions and unate cascade functions, \textquotedblright IEEE
Trans. Comput. 28 (1) (1979) 66-72.


\bibitem {Shm2}Ilya Shmulevich,
\newblock\textquotedblleft Average sensitivity of typical monotone Boolean functions,
\textquotedblright arXiv:math/0507030v1 [math. Co] 1 July 2005.

\bibitem {Shm}Ilya Shmulevich and Stuart A. Kauffman,
\newblock\textquotedblleft Activities and sensitivities in Boolean network models,
\textquotedblright Physical Review Letters, Vol 93, Number 4 (2004), 048701.

\bibitem {Shp}Igor. E. Shparlinski,
\newblock\textquotedblleft Bounds on the Fourier coefficients of the weighted sum function,
\textquotedblright Information Processing Letters. 103 (2007), 83-87.


\bibitem {Sch} Steffen Schober and Martin Bossert,
\newblock\textquotedblleft Analysis of random Boolean networks using the average sensitivity,
\textquotedblright arXiv: 0704.0197v1 [nlin.CG] 2 Apr 2007.

\bibitem {Wad}C. H. Waddington, \newblock\textquotedblleft Canalisation of
development and the inheritance of acquired characters, \textquotedblright
Nature, 150 (1942), pp. 563–564.

\bibitem {Vir}Madars Virza,
 \newblock\textquotedblleft
Sensitivity versus block sensitivity of Boolean functions,
\textquotedblright arXiv:1008.0521v2 [cs.CC]8 Dec 2010.


\bibitem {Zha} Shengyu Zhang,
\newblock\textquotedblleft Note on the average sensitivity of monotone Boolean functions,
\textquotedblright Preprint 2011.

\end{thebibliography}
\end{document}